\documentclass[11pt]{article}
\usepackage[a4paper, total={7in, 8in}]{geometry}
\usepackage{enumerate}
\usepackage{amsmath}
\usepackage{amssymb,latexsym}
\usepackage{amsthm}
\usepackage{color}
\usepackage{cancel}
\usepackage{graphicx}
\usepackage{cite}
\usepackage{longtable}
\usepackage{amscd}
\usepackage{lineno}

\makeatletter
\usepackage{comment}
\let\wfs@comment@comment\comment
\let\comment\@undefined

\usepackage{changes}
\let\wfs@changes@comment\comment
\let\comment\@undefined

\newcommand\comment{%
    \ifthenelse{\equal{\@currenvir}{comment}}
    {\wfs@comment@comment}
    {\wfs@changes@comment}%
}

\definechangesauthor[name=Daniele, color=red]{DAN}

\newtheorem{thm}{Theorem}[section]
\newtheorem{lem}[thm]{Lemma}
\newtheorem{prop}[thm]{Proposition}

\newtheorem*{Main}{Main Theorem}

\newtheorem{crit}[thm]{Criterion}

\newtheorem{defn}[thm]{Definition}
\newtheorem{rem}[thm]{Remark}

\title{On $3$-dimensional MRD codes of type $\langle x^{q^t},x+\delta x^{q^{2t}},G(x) \rangle$}
\author{Daniele Bartoli\thanks{Dipartimento di Matematica e Informatica, Universit\`a degli Studi di Perugia,  Perugia, Italy. daniele.bartoli@unipg.it}, and Francesco Ghiandoni\thanks{Dipartimento di Matematica e Informatica ``Ulisse Dini", Universit\`a  degli studi di Firenze, Firenze, Italy, francesco.ghiandoni@unifi.it}}
\begin{document}

\maketitle
\begin{abstract}
    In this work we present  results on the classification of  $\mathbb{F}_{q^n}$-linear MRD codes of dimension three. In particular, using connections with certain algebraic varieties  over finite fields, we provide non-existence results for MRD codes $\mathcal{C}=\langle x^{q^t}, F(x), G(x) \rangle \subseteq \mathcal{L}_{n,q}$ of exceptional type, i.e. such that $\mathcal{C}$  is MRD over infinite many extensions of the field $\mathbb{F}_{q^n}$.
These results  partially address a conjecture of Bartoli, Zini and Zullo in 2023.
\end{abstract}
\section{Introduction}
Let $q$ be a prime power, $n$ be a positive integer, and
denote by $\mathbb{F}_{q^n}$ the finite field with $q^n$ elements and by $\mathbb{P}^N(\mathbb{K})$ (resp.
$\mathbb{A}^N(\mathbb{K})$)  the $N$-dimensional projective (resp. affine)
space over the field $\mathbb{K}.$\\
Let $\mathcal{L}_{n,q}= \{\sum_{i=0}^{n-1}
 a_iX^{q^i} : a_i \in \mathbb{F}_{q^n} \}$ denote the $\mathbb{F}_{q}$-algebra of the $\mathbb{F}_{q}$-linearized polynomials (or $q$-polynomials) of
$q$-degree smaller than $n.$ For any $f(x) =
\sum_{i=0}^{n-1}a_iX^{q^i},$ 
we define $\deg_q(f(X)) = \max\{i: a_i \neq 0\}$ and $\min\deg_q(f(X)) = \min\{i: a_i \neq 0\}.$ We identify a polynomial $g(X) \in \mathcal{L}_{n,q}$ with the $\mathbb{F}_{q}$-linear map $x \mapsto g(x)$ over $\mathbb{F}_{q^n};$ in this way, $\mathbb{F}_{q}$-linearized polynomials over $\mathbb{F}_{q^n}$ are in one-to-one correspondence with $\mathbb{F}_q$-linear maps over $\mathbb{F}_{q^n}.$ 
\\
The $rank$ $metric$ on the $\mathbb{F}_q$-vector space $\mathbb{F}^{m\times n}_{q}$ is defined by \[d(A,B):=\textnormal{rank}(A-B) \ \ \textnormal{for} \ A,B \in \mathbb{F}^{m\times n}_{q} . \] We call a subset of $\mathbb{F}^{m\times n}_{q}$ equipped with the rank metric a $rank$-$metric$ $code.$ For a rank-metric code $\mathcal{C}$ containing at last two elements, its $minimum$ $distance$ is given by \[d(\mathcal{C}):= \min_{A,B\in \mathcal{C}, A \neq B} d(A,B).\]
When $\mathcal{C}$ is an $\mathbb{F}_q$-subspace of $\mathbb{F}^{m\times n}_{q},$ we say that $\mathcal{C}$ is an $\mathbb{F}_q$-$linear$ code of dimension $\dim_{\mathbb{F}_q}(\mathcal{C}).$
Under the assumption that $m \leq n,$ it is well known (and easily verified)
that every rank-metric code $\mathcal{C}$ in $\mathbb{F}^{m\times n}_{q}$ with minimum distance $d$ satisfies the Singleton-like  bound \[|\mathcal{C}| \leq q^{n(m-d+1)}.\]
In case of equality, $\mathcal{C}$ is called a $maximum$ $rank$-$metric$ code, or MRD code for short. MRD codes have been studied since the 1970s by Delsarte \cite{delsarte} and Gabidulin \cite{gabidulin} and have seen much interest in recent years due to an important application in network coding and cryptography \cite{koetter}. \\
From a different perspective, rank-metric codes can also be seen as sets of (restrictions of) $\mathbb{F}_q$-linear
homomorphisms from $(\mathbb{F}_{q^n})^m
$ to $\mathbb{F}_{q^n}$
 equipped with the rank metric; see \cite[Sections 2.2 and 2.3]{articolosequenze}.
In this case, it is evident that multivariate linearized polynomials can be seen
as the natural algebraic counterpart of rank-metric codes.
In particular, when $m=n,$ a rank
metric code $\mathcal{C}$ can be seen as set of $\mathbb{F}_q$-linear endomorphisms of $\mathbb{F}_{q^n},$ i.e. $\mathcal{C} \subseteq \mathcal{L}_{n,q}.$ From now on we will consider $m=n$ and $d<n.$
In the case of univariate linearized polynomials, Sheekey pointed
out in \cite{sheekey} the following connection between  $\mathbb{F}_{q^n}$-linear MRD codes and the so called $scattered$ polynomials: $\mathcal{C}_{f} = \langle x
, f(x)\rangle_{\mathbb{F}_{q^n}}$ is an MRD code with
$\dim_{\mathbb{F}_{q^n}} (\mathcal{C}) = 2$ if and only if
\[\dim_{\mathbb{F}_q}\ker(f(x)-mx) \leq 1 \] for every $m \in \mathbb{F}_{q^n}.$ The concept of scattered polynomial  introduced in \cite{sheekey}
has been slightly generalized in \cite{bartoli Ip eq 0 o m}.
\begin{defn}\cite{bartoli Ip eq 0 o m, sheekey}
An $\mathbb{F}_q$-linearized polynomial $f(X)\in \mathbb{F}_{q^n}[X]$ is called a \textbf{scattered} polynomial of index $t \in \{0,\dots,n-1\}$ if  \[\dim_{\mathbb{F}_q}\ker(f(x)-mx^{q^t}) \leq 1, \] for every $m \in \mathbb{F}_{q^n}.$ Also, a scattered polynomial of index $t$ is \textbf{exceptional} if it is scattered of index $t$ over infinitely many extensions $\mathbb{F}_{q^{nm}}$ of $\mathbb{F}_{q^{n}}.$ 
    \end{defn}
While several families of scattered polynomials have been
constructed in recent years \cite{bartolizanellazullo, BlokhuisLavrauw, CsajbókMarinoPolverinoZanella, CsajbókMarinoZullo, LongobardiZanella, LunardonTrombettiZhou, NeriSantonastasoZullo, sheekey, zanella}, only two families of
exceptional ones are known:
\begin{itemize}
    \item [(Ps)] $f(x)=x^{q^t}$ of index $0,$ with $\gcd(t,n)=1$ (polynomials of so-called pseudoregulus type);
    \item [(LP)] $f(x)=x+\delta x^{q^{2t}}$ of index $t,$ with $\gcd(t,n)=1$ and $N_{q^n/q}(\delta)=\delta^{(q^n-1)/(q-1)}\neq 1$ (so called LP polynomials).
\end{itemize}
From a coding theory point of view, if $f$ is exceptional scattered of index $t,$ the corresponding rank distance code $\mathcal{C}^m_{f,t} = \langle x^{q^t}
, f(x)\rangle_{\mathbb{F}_{q^{mn}}}\subseteq \mathcal{L}_{nm,q}$ turns out to be an MRD code for infinitely many $m;$ codes of this kind are called exceptional $\mathbb{F}_{q^n}$-linear MRD codes (see \cite{bartolizinizullo}). 
Moreover, in \cite{articolosequenze} the authors  introduce the  notions of $h$-scattered sequences and exceptional $h$-scattered sequences which parameterize exceptional MRD codes.
Only two families of exceptional $\mathbb{F}_{q^n}$-linear MRD codes are
known so far:
\begin{itemize}
    \item [(G)] $\mathcal{G}_{r,s} = \langle x, x^{q^s}
, \dots , x^{q^{s(r-1)}} \rangle_{\mathbb{F}_{q^n}},$  with $\gcd(s,n) = 1,$
see \cite{delsarte,gabidulin};
       \item [(T)] $\mathcal{H}_{r,s}(\delta) = \langle x^{q^s}
,\dots, x^{q^{s(r-1)}}
, x + \delta x^{q^{sr}}
\rangle_{\mathbb{F}_{q^n}},$ with
$\gcd(s,n) = 1$ and $N_{q^n / q}(\delta) \neq (-1)^{nr}$ see \cite{LunardonTrombettiZhou,sheekey}.
\end{itemize}
    
The first family is known as generalized Gabidulin codes and
the second one as generalized twisted Gabidulin codes. \\
In \cite{bartolizhou} it has been shown
that the only exceptional $\mathbb{F}_{q^n}$-linear MRD codes spanned by
monomials are the codes (G), in connection with so-called
Moore exponent sets, while in \cite{bartolizinizullo} the authors investigated exceptional $\mathbb{F}_{q^n}$-linear
MRD codes not generated by monomials and proved that an exceptional $r$-dimensional $\mathbb{F}_{q^n}$-linear MRD code 
contains an exceptional scattered polynomial (see Theorem \ref{thm bartolizinizullo}). \\
 Motivated by this last necessary condition on MRD codes of exceptional type, we considered codes of type $\mathcal{C}=\langle x^{q^t}, F(x), G(x) \rangle_{\mathbb{F}_{q^n}}$ and we address a conjecture in \cite{bartolizinizullo} for $r=3$ and $F(x)$ a LP polynomial.
 The techniques that we use to prove our results are from algebraic geometry over finite fields. This approach has already proved successful in the investigation of families of functions with many applications in
 coding theory and cryptography; see for instance \cite{aubry mcguire rodier 7,BartoliFG,blmt,BartoliTimpanella,hernando mcguire,janwa mguire wilson,jedlicka,TZ}.
Our main result can be summarized as follows (see Theorem \ref{thm finale}).
\begin{Main} 
    If $(t,q)\notin \{(1,3);(1,4);(1,5);(2,3);(2,4);(2,5),(4,3)\},$ then there are no exceptional $3$-dimensional $\mathbb{F}_{q^n}$-linear MRD codes of type $\mathcal{C}=\langle x^{q^t}, x+\delta x^{q^{2t}}, G(x) \rangle \subseteq \mathcal{L}_{n,q},$ with   $\deg_q(G(x))>2t.$
\end{Main}

\section{Preliminaries on algebraic curves and varieties}
Let $F(X,Y)\in \mathbb{K}[X,Y]$, $\mathbb{K}$ a field, be a polynomial defining an affine plane curve $\mathcal{C}: F(X,Y)=0$. A plane curve is absolutely irreducible if there are no non-trivial factorizations of its defining polynomial $F(X,Y)$ in $\overline{\mathbb{K}}[X,Y]$, where $\overline{\mathbb{K}}$ is the algebraic closure of $\mathbb{K}$. If $F(X,Y)=\prod_i F^{(i)}(X,Y)$, with $F^{(i)}(X,Y)\in \overline{\mathbb{K}}[X,Y]$ of positive degree, then $\mathcal{C}_i: F^{(i)}(X,Y)=0$ are called components of $\mathcal{C}$. A component is $\mathbb{F}_q$-rational if  it is fixed by the Frobenius morphism $\varphi$ or equivalently  $\lambda F^{(i)}(X,Y)\in \mathbb{K}[X,Y]$ for some $\lambda \in \overline{\mathbb{K}}$. 

Let $P=(u,v)\in \mathbb{A}^2(\mathbb{K})$ be a point in the plane, and write

\[
F(X+u,Y+v)=F_0(X,Y)+F_1(X,Y)+F_2(X,Y)+\cdots,
\]
where $F_i$ is either zero or homogeneous of degree $i$. The \emph{multiplicity} of $P\in \mathcal{C}$, written as $m_P(\mathcal{C})$ or $m_P(F)$, is the smallest integer $m$ such that $F_m\ne 0$ and $F_i=0$ for $i<m$;  $F_m=0$ is the \emph{tangent cone} of $\mathcal{C}$ at $P$. A linear component of the tangent cone is called a \emph{tangent} of $\mathcal{C}$ at $P$. The point $P$ is on the curve $\mathcal{C}$ if and only if $m_P(\mathcal{C})\ge 1$. If $P$ is on $\mathcal{C}$, then $P$ is a \emph{simple} point of $\mathcal{C}$ if $m_P(\mathcal{C})=1$, otherwise $P$ is a \emph{singular} point of $\mathcal{C}$. It is possible to define in a similar way the multiplicity of an ideal point of $\mathcal{C}$, that is a point of the curve lying on the line at infinity. We denote by $Sing(\mathcal{C})$ the set of singular points of the curve $\mathcal{C}$.

Given two plane curves $\mathcal{A}$ and $\mathcal{B}$ and a point $P$ on the plane, the \emph{intersection number} (or \emph{intersection multiplicity})  $I(P, \mathcal{A} \cap \mathcal{B})$ of $\mathcal{A}$ and $\mathcal{B}$ at the point $P$ can be defined by seven axioms. We do not include its precise and long definition here. For more details, we refer to \cite{fulton} and \cite{HKT} where the intersection number is defined equivalently in terms of local rings and in terms of resultants, respectively.

For a given plane curve $\mathcal{C}$ and a point $P\in \mathcal{C}$, we denote by $I_{P,max}(\mathcal{C})$ the maximum
possible intersection multiplicity of two components of  $\mathcal{C}$ at $P \in Sing(\mathcal{C}).$ 
We list here two useful results in this direction.


\begin{lem}  \cite[Section 3.3]{fulton}   \label{lemma easy} \cite[Lemma 4.3]{schmidt Ip eq zero o m} \cite[Lemma 2.5]{bartoli Ip eq 0 o m} Let  $q$ be a prime power and $F(X,Y)\in \mathbb{F}_{q}[X,Y]$. Let  $P=(\alpha,\beta)\in \mathbb{F}_{q}^{2}$ and write  \[F(X+\alpha,Y+\beta)=F_{m}(X,Y)+F_{m+1}(X,Y)+\dots,\] where $F_{i}\in \mathbb{F}_{q}[X,Y]$ is   zero or homogeneous of degree  $i$ and  $F_{m} \neq 0.$ The following properties hold. \begin{itemize}
	    \item[(i)] If  $F_{m}(X,Y)$ is separable, then $I_{P,max}(\mathcal{C})\leq \lfloor m^2/2\rfloor.$
     \item[(ii)] Suppose that  $F_m=L^m$ with  $L$  a linear  form; 
     \begin{itemize}
         \item if $L \nmid F_{m+1}$ then $I_{P,max}(\mathcal{C})=0,$
         \item if 
	$L^2 \nmid F_{m+1} $  then  $I_{P,max}(\mathcal{C})\leq m.$
     \end{itemize}
	\end{itemize} 
\end{lem}

\begin{crit}\label{criterio due noni}  \cite[Lemma 11]{jedlicka}
Let	  $\mathcal{C}: h(X,Y)=0$   be a curve of degree $n$  defined over $\mathbb{F}_{q}$.

    If 	$ \displaystyle\sum_{P \in Sing(h)}I_{P,max}(\mathcal{C}) < \frac{2}{9}\deg^{2}(h)  $  \ \ \
 then  $\mathcal{C}$   possesses at least one absolutely irreducible component defined over $\mathbb{F}_{q}.$
	\end{crit}

Consider the set $\overline{\mathbb{F}_{q}}[[t]]$ of the formal power
series on $t.$ Let $(x_0, y_0)\in \overline{\mathbb{F}_{q}}^2$ be an affine point of $\mathcal{C}: F(X, Y)= 0.$ A $branch$ of center
$(x_0, y_0)$ of $\mathcal{C}$ is a point $(x(t), y(t)) \in (\overline{\mathbb{F}_{q}}[[t]])^2$ such that $F(x(t), y(t)) = 0,$ where
\begin{eqnarray*}
 x(t) &=& x_0 + u_1t + u_2t^2 + \dots, \\
y(t) &=& y_0 + v_1t + v_2t^2 + \dots.   
\end{eqnarray*}
See \cite[Chapter 4]{HKT} for more details on branches. There exists a unique branch centered
at a simple point of $\mathcal{C}.$ If there exists only a branch centered in a point $P\in \mathcal{C}$ then $I_{P,max}(\mathcal{C})=0.$ 

\begin{rem}
    In order to determine branches centered at singular points of a curve $\mathcal{C}$ we
make use of quadratic transformations; see \cite[Section 4]{HKT}. Consider
a curve $\mathcal{C}$ defined by $$F(X, Y ) = F_r(X, Y) + F_{r+1}(X, Y) + \dots = 0,$$ where each $F_i(X, Y)$
is homogeneous in $X$ and $Y$ and of degree $i$ and $F_r \neq 0.$ First, we can suppose that the singular
point under examination is the origin $O = (0, 0)$ and that $X = 0$ is not a tangent line at $O.$ Let $r$ be its multiplicity. The geometric transform of a curve $\mathcal{C}$ is the curve $\mathcal{C}'$ given
by $F'(X, Y) = F(X, XY)/X^r.$ (If $Y = 0$ is not a tangent line at $O$ then we can also
consider $\mathcal{C}'$ defined by $F'
(X, Y) = F(XY,Y)/Y^r).$ By \cite[Theorem 4.44]{HKT}, there
exists a bijection between the branches of $\mathcal{C}$ centered at the origin and the branches of
$\mathcal{C}'$ centered at an affine point on $X = 0.$ In our proofs we will perform chains of local
quadratic transformations until the total number of branches is determined. In particular,
if the tangent cone $F_r(X, Y )$
at $O$ splits into non-repeated linear factors (over the algebraic closure), then there are precisely $r$ distinct branches centered at $O.$ In fact, distinct linear factors
of $F_r(X, Y )$ correspond to distinct affine points of $\mathcal{C}'$ on $X = 0.$
   \end{rem}

An algebraic hypersurface is an algebraic variety that may be defined by a single polynomial equation. An algebraic
hypersurface defined over a field $\mathbb{K}$ is $absolutely$ $irreducible$
if the associated polynomial is irreducible over every algebraic extension of $\mathbb{K}.$ An absolutely irreducible $\mathbb{K}$-rational
component of a hypersurface $\mathcal{V},$ defined by the polynomial
$F,$ is simply an absolutely irreducible hypersurface which is
associated to a non-costant factor of $F$ defined over $\mathbb{K}.$ 
\begin{lem} \label{sette Lemma 2.1} \cite[Lemma 2.1]{aubry mcguire rodier 7}
	Let $X,H \subseteq \mathbb{P}^{N}(\mathbb{F}_{q})$ be   projective hypersurfaces. If  $X \cap H$  has a non-repeated absolutely irreducible component  defined over  $\mathbb{F}_{q}$, then  $X$  has an absolutely irreducible component  defined over $\mathbb{F}_{q}.$
\end{lem}
In our investigation we will need bounds on the number of $\mathbb{F}_{q}$-rational points of algebraic varieties and we will make use of the following result a number of times.

\begin{thm}\cite[Theorem 7.1]{MR2206396} \label{thm cafure matera}
Let $\mathcal{W}$ be an absolutely irreducible variety defined over $\mathbb{F}_q$ of dimension $n$ and degree $d$. If $q>2(n+1)d^2$, then 
$$\#(\mathcal{W}\cap \mathbb{A}^N(\mathbb{F}_q))\geq q^n-(d-1)(d-2)q^{n-1/2}+5d^{13/3} q^{n-1}.$$
\end{thm}

\section{Scattered sequences of order $1$ and MRD codes}
In this section we recall the notion of  scatteredness for subspaces and sequences in $\mathbb{F}_{q^n}^r,$
and how they are related to rank-metric codes.

\begin{defn}\cite{CMPZcombinatorica}
    Let $h, r,  n$ be positive integers, such that $h < r.$ An $\mathbb{F}_q$-subspace
$U \subseteq \mathbb{F}_{q^n}^r$  is said to be \textbf{$h$-scattered} if for every $h$-dimensional $\mathbb{F}_{q^n}$-subspace $H \subseteq \mathbb{F}_{q^n}^r,$
it holds $\dim_{\mathbb{F}_{q}}
(U \cap H) \leq h.$ When $h = 1,$ a $1$-scattered subspace is simply called \textbf{scattered}.
\end{defn}
For what concerns $h$-scattered subspaces, there is a well-known bound on their $\mathbb{F}_{q}$-dimension.
Namely, an $h$-scattered subspace $U \subseteq \mathbb{F}_{q^n}^r,$ if $U$ does not define a subgeometry, satisfies
\begin{equation} \label{eq upperbound hscattered}
   \dim_{\mathbb{F}_{q}}
(U ) \leq \dfrac{rn}{h+1};
\end{equation}

see \cite{CMPZcombinatorica}. An $h$-scattered subspace meeting (\ref{eq upperbound hscattered}) with equality is called a \textbf{maximum $h$-scattered subspace}.\\
Let $\mathcal{G} = \{g_1, \dots , g_k\}\subseteq \mathcal{L}_{n,q}[X_1,\dots,X_m],$ where $\mathcal{L}_{n,q}[X_1,\dots,X_m] :=\left\{\displaystyle\sum_{i=1}^{m}\displaystyle\sum_{j=0}^{n-1}\gamma_{i,j}X_i^{q^j} \ : \ \gamma_{i,j} \in \mathbb{F}_{q^n}\right\},$ and consider the $\mathbb{F}_q$-space
\begin{equation}
    U_{\mathcal{G}} := \{(g_1(x_1, \dots , x_m), \dots , g_r(x_1,\dots , x_m)) : x_1, \dots , x_m \in \mathbb{F}_{q^n}
 \} \subseteq  \mathbb{F}_{q^n}^r.
\end{equation}
\begin{defn} \cite{articolosequenze}
Let $\mathcal{I} := (i_1, i_2,\dots, i_m) \in (\mathbb{Z}/n\mathbb{Z})^m$
 and consider $f_1, \dots , f_r \in \mathcal{L}_{n,q}[X_1,\dots,X_m].$ \\ Let $U_{\mathcal{I},\mathcal{F}} := U_{\mathcal{F}'},$ where
$\mathcal{F}'= (X_1^{q^{i_1}},\dots, X_m^{q^{i_m}}, f_1, \dots , f_s)\subseteq \mathcal{L}_{n,q}[X_1,\dots,X_m].$ \\
The $s$-tuple $\mathcal{F} := (f_1, \dots , f_s)$ is said to be an $(\mathcal{I}; h)_{q^n}$-\textbf{scattered sequence} of order $m$ if  $U_{\mathcal{I},\mathcal{F}}$ is maximum $h$-scattered in $\mathbb{F}_{q^n}^{m+s}.$   \\
An $(\mathcal{I}; h)_{q^n}$-scattered sequence $\mathcal{F} := (f_1, \dots , f_s)$ of order $m$ is said to be
\textbf{exceptional} if it is $h$-scattered over infinitely many extensions $\mathbb{F}_{q^n}^{\ell}$ of $\mathbb{F}_{q^n}.$
\end{defn}
As the following remark shows, $(\mathcal{I}; h)_{q^n}$-scattered sequences, with $|\mathcal{I}|=1,$  have been considered also in \cite{bartolizinizullo}, though with slightly different terminology.  
\begin{rem} \label{rem link scatt seq and pol set}
    It is not difficult to see that, for $m=1$ and $\mathcal{I}=\{t\},$  an $(r-1)$-tuple $(f_2,\dots,f_r) \subseteq \mathcal{L}_{n,q}$ is a $(\mathcal{I}; r-1)_{q^n}$-scattered  (or simply $(t,r-1)_{q^n}$-scattered) sequence of order $1$ if and only if, for any $\alpha_1,\dots,\alpha_r \in \mathbb{F}_{q^n},$ 
  \begin{equation*}
        \det \begin{pmatrix}
  \alpha_1^{q^t} & f_2(\alpha_1) & \cdots  &  f_r(\alpha_1)\\
\alpha_2^{q^t} & f_2(\alpha_2) & \cdots  &  f_r(\alpha_2) \\
\vdots & \vdots & \cdots & \vdots \\
 \alpha_r^{q^t}  & f_2(\alpha_r) & \cdots  &  f_r(\alpha_r)
\end{pmatrix}=0 \ \ \ \ \Longrightarrow \ \  \dim_{\mathbb{F}_q}\langle \alpha_1,\dots,\alpha_r \rangle_{\mathbb{F}_q} < r;
    \end{equation*}
    see for istance \cite{bartolizhou,bartolizinizullo,zanella} for an explicit link between scattered spaces and Moore matrices.\\
    If the previous property holds, $\underline{f}=(x^{q^t},f_2,\dots,f_r)$ is said to be a \textbf{Moore polynomial set} for $q$ and $n$ of index $t$ (see \cite[Definition 9]{bartolizinizullo}).
\end{rem}


Moore polynomial sets can be characterized in terms of MRD codes as follows.

\begin{thm} \label{thm collegamento MRD moore set}\cite{bartolizinizullo}
    Let $r$ and $n$ be positive integers with $r \leq n+1,$ and let $\underline{f}=(x^{q^t},f_2(x),\dots,f_r(x)),$ where $x^{q^t},f_2(x),\dots,f_r(x) \in \mathcal{L}_{n,q}$ are $\mathbb{F}_{q^n}$-linearly independent. The $\mathbb{F}_{q^n}$-linear rank metric code \[\mathcal{C}_{\underline{f}}=\langle x^{q^t},f_2(x),\dots,f_r(x)\rangle_{\mathbb{F}_{q^n}}\]
    is an MRD code if and only if $\underline{f}$ is a Moore polynomial set for $q$ and $n.$
\end{thm}

Now we focus on the exceptionality of $\mathbb{F}_{q^n}$-linear MRD codes $\mathcal{C} \subseteq \mathcal{L}_{n,q}$ of dimension $r,$ or equivalently by Theorem  \ref{thm collegamento MRD moore set}, on the exceptionality of scattered sequences of order 1.
We can assume infact, without restrictions, the following properties on the polynomials generating a non-degenerate $\mathbb{F}_{q^n}$-linear code $\mathcal{C}$ (for details see \cite{bartolizinizullo} and \cite[Definition 2.8]{articolosequenze}). In particular, from \cite[Proposition 2.11]{articolosequenze} and \cite[Corollary IV.10]{pssz divisible codes} we can assume that $\mathcal{C}$ contains a monomial.
\begin{rem}\cite[Properties 13] {bartolizinizullo}  Given a non-degenerate $\mathbb{F}_{q^n}$-linear code $\mathcal{C}$ of dimension $r,$ there exist $f_1(x),\dots,f_r(x) \in \mathcal{C}$ such that the following properties hold:
    \begin{itemize}
        \item [(1)] $f_1(x)=x^{q^t};$
        \item [(2)] $f_1(x),\dots,f_r(x)$ are $\mathbb{F}_{q^n}$-linearly independent;
        \item [(3)] $M_1:=\deg_q(f_1(x)), \dots, M_r:=\deg_q(f_r(x))$ are all distinct;
        \item [(4)] $m_1:=\textnormal{min}\deg_q(f_1(x)),\dots,m_r:=\textnormal{min}\deg_q(f_r(x))$ are all distinct, and $m_i=0$ for some $i;$
        \item [(5)] $f_1(x),\dots,f_r(x)$ are monic;
        \item [(6)] for any $i,$ if $f_i(x)$ is a monomial then $m_i=M_i \geq t.$
    \end{itemize}
\end{rem}
 A Moore polynomial set $\underline{f}=(f_1(x),\dots,f_r(x))\subseteq \mathcal{L}_{n,q}$ satisfying the previous six properties is said to be a Moore polynomial set for $q$ and $n$ of index $t.$ 
\section{Scattered sequences and algebraic varieties}
In this section, we consider varieties introduced in \cite{bartolizinizullo}  to traslate the determination of scattered sequences of order $1$ into an algebraic geometry problem. 
\begin{itemize}
    \item $\mathcal{U}:=\mathcal{U}_{\underline{f}} \subset \mathbb{P}^r(\overline{\mathbb{F}_{q^n}}),  \ \ \ \ \  \mathcal{U}: F_{\underline{f}}(X_1,\dots,X_r):=\det(M_{\underline{f}}(X_1,\dots,X_r))=0,$ \\
where
\begin{equation*}
    M_{\underline{f}}(X_1,\dots,X_r)=\begin{pmatrix}
  f_1(X_1) & f_2(X_1) & \cdots  &  f_r(X_1)\\
f_1(X_2) & f_2(X_2) & \cdots  &  f_r(X_2) \\
\vdots & \vdots & \cdots & \vdots \\
f_1(X_r) & f_2(X_r) & \cdots  &  f_r(X_r)
\end{pmatrix};
\end{equation*}
\item $\mathcal{V}:=\mathcal{U}_{(x,x^q,\dots,x^{q^{r-1}})}\subset \mathbb{P}^r(\overline{\mathbb{F}_{q^n}}), \ \ \ \ \ \ \ \mathcal{V}: F_{(x,x^q,\dots,x^{q^{r-1}})}(X_1,\dots,X_r)=0$ \\ where
\begin{equation} \label{eq denominatore hypersurfaces}  F_{(x,x^q,\dots,x^{q^{r-1}})}(X_1,\dots,X_r)=\prod_{(a_1,\dots,a_r)\in \mathbb{P}^{r-1}(\mathbb{F}_q)} (a_1X_1+\dots +a_rX_r); \end{equation}
\item $\mathcal{W}\subset \mathbb{P}^r(\overline{\mathbb{F}_{q^n}}), \ $  with affine equation
\begin{equation} \label{equation W}\mathcal{W}: \dfrac{F_{\underline{f}}(X_1,\dots,X_r)}{F_{(x,x^q,\dots,x^{q^{r-1}})}(X_1,\dots,X_r)}=0.\end{equation}
\end{itemize}
The link between scattered sequence of order $1$ and algebraic hypersurfaces is straightforward.
\begin{prop}  \label{prop collegamento varietà}\cite{bartolizinizullo}
The $(r-1)$-tuple $(f_2,\dots,f_r)$ is a $(\{t\},r-1)_{q^n}$-scattered sequence of order $1$ if and only if all the affine $\mathbb{F}_{q^n}$-rational points of $\mathcal{W}$ lie on $\mathcal{V}.$ 
\end{prop}
\begin{thm} \cite[Main Theorem] {bartolizinizullo}  \label{thm bartolizinizullo}
    Let $\mathcal{C}\subseteq \mathcal{L}_{n,q}$ be an exceptional $r$-dimensional 
 $\mathbb{F}_{q^n}$-linear MRD code containing at least a separable polynomial $f(x)$ and a monomial. If $r>3,$ assume also that $q>5.$ Let $t$ be the minimum integer such that $x^{q^t} \in \mathcal{C}.$ \\
    If $t>0$ and $\mathcal{C}=\langle x^{q^t},f(x),g_3(x),\dots,g_r(x) \rangle_{\mathbb{F}_{q^n}},$ with $\deg(g_i(x)) > \max\{q^t,\deg(f(x))\}$ for each $i=3,\dots,r,$ then $f(x)$ is exceptional scattered of index $t.$
    
\end{thm}
\begin{rem}
    Note that without loss of generality, we can always assume that a code $\mathcal{C}$ contains a separable polynomial ($\!\!$\cite[Remark 12]{bartolizinizullo}), and if $\mathcal{C}$ is non-degenerate it also contains a monomial (see \cite[Lemma 2.1]{ltz elto invert in cod nondeg} and \cite[Corollary IV.10]{pssz divisible codes}).
\end{rem}
As previously stated in the introduction, until today, the only known non-monomial example of exceptional scattered polynomials, for arbitrary $t,$ is given by the LP polynomials; therefore it is natural to search for exceptional MRD codes where $f(x)$ is of such type. In the following we will focus on RD codes of dimension $3.$ 

\section{Moore polynomial sets of type $\underline{f}=(x^{q^t},x+\delta x^{q^{2t}},G(x))$}
In this section we investigate curves arising from Moore polynomial sets for $q$ and $n,$ of index $t,$ of type $\underline{f}=(x^{q^t},x+\delta x^{q^{2t}},G(x)).$ They correspond to $(\{t\},2)_{q^n}$-scattered sequences of order $1$ as we have seen in Remark \ref{rem link scatt seq and pol set}.
Let $q$ be a prime power and $t,n$ positive integers.
\begin{prop} \label{prop gradi minimi in progressione} \cite[Proposition 23]{bartolizinizullo}
    If the paire $(x+\delta x^{q^{2t}},G(x))\subseteq \mathcal{L}_{n,q}$ is a $(\{t\},2)_{q^n}$-scattered sequence of order $1$ and $n>4\deg_q(G)+2,$ then $\min\deg_q(G)=2t$ or $\min\deg_q(G)=t/2.$
\end{prop}
Let consider first the case $\min\deg_q(G)=2t.$ \\
 Set 
\begin{equation}  \label{eq f e g}
    F:=X+\delta X^{q^{2t}}, \ \ \ \ \  N_{q^n|q}(\delta)\neq 1, 
\end{equation}
\begin{equation}  \label{eq g caso 2t}
    G:=X^{q^{2t}} + \dots + CX^{q^k}, \ \ \ \ \ C\neq 0,
\end{equation}
where $F,G \in \mathbb{F}_{q^n}[X],$ and  $0 < 2t < k < n.$ \\
Fix an element $\lambda \in \mathbb{F}_{q^n} \setminus \mathbb{F}_q$ such that $F(\lambda)\neq 0 \neq G(\lambda)$ and \begin{equation} \label{lambda scelto opportunamente}
  L_{\xi}(\lambda):=(\xi-\xi^{q^{t}})^{q^{k+t}}(\xi^{q^{k-t}}-\xi)^{q^t}F(\lambda)^{q^t+1}+\delta(\xi^{q^k}-\xi)^{q^t(q^t+1)}\lambda^{q^t(q^t+1)}\neq 0,  
\end{equation}
for each $\xi \in \mathbb{F}_{q^{k-2t}} \setminus \mathbb{F}_{q^{k-t}}.$ Such an element exists in any field $\mathbb{F}_{q^n},$ with $n\geq k+t+1.$ Indeed the polynomial $L_\xi$ is not zero and of degree $q^{3t}+q^{2t},$ for each $\xi \in \mathbb{F}_{q^{k-2t}} \setminus \mathbb{F}_{q^{k-t}},$ so $$\#\left(\bigcup_{\xi \in \mathbb{F}_{q^{k-2t}}}\!\!\{\eta : L_\xi(\eta)=0\}\right) \leq (q^{3t}+q^{2t})q^{k-2t} = q^{k+t}+q^k < q^n$$ for $n\geq k+t+1.$ \\
Consider the curves 
\begin{equation} \label{eq curva C}
    \mathcal{C:}\begin{vmatrix}
X^{{q^t}} & F(X) & G(X)\\
Y^{{q^t}} & F(Y) & G(Y) \\
\lambda^{{q^t}} & F(\lambda) & G(\lambda)
\end{vmatrix}=0,
\end{equation}
\begin{equation}  \label{eq curva A}
    \mathcal{A:}\dfrac{\begin{vmatrix}
X^{{q^t}} & F(X) & G(X)\\
Y^{{q^t}} & F(Y) & G(Y) \\
\lambda^{{q^t}} & F(\lambda) & G(\lambda)
\end{vmatrix}}{\begin{vmatrix}
X & X^{{q}} & X^{{q^2}}\\
Y & Y^{{q}} & Y^{{q^2}} \\
\lambda & \lambda^{{q}} & \lambda^{{q^2}}
\end{vmatrix}}=0.
\end{equation}
\begin{lem} \cite[Lemma 18]{bartolizinizullo}\label{lem specializzo z in lambda}
    Let $\mathcal{V}:F_{(x,x^q,x^{q^2})}(X,Y,Z)=0$ and $\mathcal{W}: \dfrac{F_{(x^{q^t},x+\delta x^{q^{2t}},G(x))}(X,Y,Z)}{F_{(x,x^q,x^{q^2})}(X,Y,Z)}=0$ as in Equations \ref{eq denominatore hypersurfaces} and \ref{equation W}. 
    If $\mathcal{A}$ has a non-repeated $\mathbb{F}_{q^n}$-rational absolutely irreducible component not contained in the curve defined by  $F_{(x,x^q,x^{q^2})}(X,Y,\lambda)=0,$ then $\mathcal{W}$ has a non-repeated $\mathbb{F}_{q^n}$-rational absolutely irreducible component not contained in $\mathcal{V}.$ 
\end{lem}
In the following we exploit a method introduced in \cite{janwa mguire wilson} via the investigation of singular points of the curves
$ \mathcal{A}$ and $\mathcal{C},$
see Criterion \ref{criterio due noni}.

In particular, our aim is to prove the existence of an absolutely irreducible $\mathbb{F}_{q^n}$-rational component  of the curve $\mathcal{A}$. We proceed as follows.
\begin{enumerate}
    \item We determine the set $Sing({\mathcal{A}})$ of singular points of ${\mathcal{A}}$; see below.
    \item For each point $P\in Sing({\mathcal{A}})$ we provide upperbounds on $I_{P,max}(\mathcal{A})$; see Propositions \ref{prop stime ipmax caso affine}, \ref{prop stime ipmax punti all'infinito}, Lemma \ref{lem un solo ramo} and Proposition \ref{prop caso t mezzi ipmax singolarità infinito}.
    \item We compute an upper bound on $\sum_{P\in Sing(\mathcal{A})}I_{P,max}(\mathcal{A})$ and obtain the desired result via Criterion \ref{criterio due noni}; see Theorems  \ref{th main caso 2t} and \ref{thm A componente caso t mezzi}.
\end{enumerate}
We will investigate the singular points of $\mathcal{C};$ in fact, as it can be easily seen, the set of its singular points contains also the singular points of $\mathcal{A}.$  \\
Note that an affine point $P=(\overline{x},\overline{y}) \in \mathcal{C}$ is singular if and only if
\begin{equation*}
\begin{vmatrix}
\overline{x}^{{q^t}} & G(\overline{x})\\
\lambda^{{q^t}} &  G(\lambda)
\end{vmatrix}=\begin{vmatrix}
\overline{y}^{{q^t}} & G(\overline{y})\\
\lambda^{{q^t}} &  G(\lambda)
\end{vmatrix}=0.
\end{equation*}
One can see immediately that
\begin{equation*}
    \begin{vmatrix}
(X+\overline{x})^{{q^t}} & F(X+\overline{x}) & G(X+\overline{x})\\
(Y+\overline{y})^{{q^t}} & F(Y+\overline{y}) & G(Y+\overline{y}) \\
\lambda^{{q^t}} & F(\lambda) & G(\lambda)
\end{vmatrix}=H_{q^t}+H_{q^t+1}+\dots,
\end{equation*}
where
\begin{eqnarray}
  H_{q^t}&=&\begin{vmatrix}
F(\overline{y}) & G(\overline{y})\\
F(\lambda) &  G(\lambda)
\end{vmatrix}X^{q^t}-\begin{vmatrix}
F(\overline{x}) & G(\overline{x})\\
F(\lambda) &  G(\lambda)
\end{vmatrix}Y^{q^t}; \\
  H_{q^t+1}&=&G(\lambda)(X^{q^t}Y-XY^{q^t}).
\end{eqnarray}
As a direct consequence of Lemma \ref{lemma easy} we have the following.
\begin{prop}  \label{prop stime ipmax caso affine}
    Let $\mathcal{C}:F_{\underline{f}}(X,Y,\lambda)=0$ and $P=(\overline{x},\overline{y})\in \mathcal{C}$ such that \begin{equation*}
\begin{vmatrix}
\overline{x}^{{q^t}} & G(\overline{x})\\
\lambda^{{q^t}} &  G(\lambda)
\end{vmatrix}=\begin{vmatrix}
\overline{y}^{{q^t}} & G(\overline{y})\\
\lambda^{{q^t}} &  G(\lambda)
\end{vmatrix}=0. 
\end{equation*}
Then \begin{itemize}
    \item $I_{P,max}(\mathcal{C}) \leq \frac{(q^t+1)^2}{4}$ $ \ \ $ if $ \ $ $\begin{vmatrix}
F(\overline{y}) & G(\overline{y})\\
F(\lambda) &  G(\lambda)
\end{vmatrix}=\begin{vmatrix}
F(\overline{x}) & G(\overline{x})\\
F(\lambda) &  G(\lambda)
\end{vmatrix}=0;$
\item $I_{P,max}(\mathcal{C})\leq q^t$ $ \ \ \ $ otherwise.
\end{itemize}    
\end{prop}

\begin{rem}  \label{rem numero di singolarità affini}
 Observe that
\begin{equation*}
\begin{vmatrix}
\overline{x}^{{q^t}} & G(\overline{x})\\
\lambda^{{q^t}} &  G(\lambda)
\end{vmatrix}=[\overline{x}G(\lambda)^{1/q^t}-\lambda(\overline{x}^{q^t}+\dots +C^{1/q^t}\overline{x}^{q^{k-t}})]^{q^t}=g(\overline{x})^{q^t},
\end{equation*}
where $g$ is a separable polynomial of degree $q^{k-t}.$ Thus the number of affine singular points of $\mathcal{C}$ is at most $q^{2(k-t)}.$ 
Moreover, if $m_P(\mathcal{C})=q^t+1,$ that is \begin{equation}
  \begin{vmatrix}
F(\overline{y}) & G(\overline{y})\\
F(\lambda) &  G(\lambda)
\end{vmatrix}=\begin{vmatrix}
F(\overline{x}) & G(\overline{x})\\
F(\lambda) &  G(\lambda)
\end{vmatrix}=0, 
\end{equation} then we obtain  \begin{equation*}
    \dfrac{\overline{x}^{q^t}}{\lambda^{q^t}}=\dfrac{G(\overline{x})}{G(\lambda)}=\dfrac{F(\overline{x})}{F(\lambda)},
\end{equation*}
i.e. \begin{equation} \label{eq F x segnato lambda}
    F(\overline{x})\lambda^{q^t}-F(\lambda)\overline{x}^{q^t}=\delta \lambda^{q^t} \overline{x}^{q^{2t}}-(\lambda +\delta \lambda^{q^{2t}})\overline{x}^{q^t} +\lambda^{q^t}\overline{x}=0.
\end{equation}
By combining Equation \ref{eq F x segnato lambda} with $g(\overline{x})=\overline{x}G(\lambda)^{1/q^t}-\lambda(\overline{x}^{q^t}+\dots +C^{1/q^t}\overline{x}^{q^{k-t}})=0,$ we get that $\overline{x}$ (or $\overline{y}$ equivalently) must be a root of a polynomial 
\begin{equation}
    h(X)^{q^t}:=[\alpha_0 X + \dots + \alpha_{k-2t} X^{q^{k-2t}}+\Tilde{\alpha}X^{q^t}]^{q^t}
\end{equation}
for suitable $\alpha_0,\dots,\alpha_{k-2t}, \Tilde{\alpha} \in \mathbb{F}_{q^n},$ where $\deg(h)\leq q^{\max\{k-2t,t\}}.$ Thus \[ \# \{ P \in \mathcal{C}: m_P(\mathcal{C})=q^t+1\}\leq q^{\max\{2k-4t,2t\}}.\]  
\end{rem}

 Consider now the line at infinity $\ell_\infty : Z=0.$ A homogeneous equation of $\mathcal{C}$ is given by 
\begin{equation}
    \begin{vmatrix}
X^{{q^t}} & XZ^{q^{2t}-1}+\delta X^{q^{2t}} & X^{q^{2t}}Z^{q^k-q^{2t}}+\dots +CX^{q^k}\\
Y^{{q^t}} & YZ^{q^{2t}-1}+\delta Y^{q^{2t}} & Y^{q^{2t}}Z^{q^k-q^{2t}}+\dots +CY^{q^k} \\
\lambda^{{q^t}} & F(\lambda)Z^{q^{2t}-q^t} & G(\lambda)Z^{q^k-q^t}
\end{vmatrix}=0.
\end{equation} Thus  $\mathcal{C}\cap \ell_\infty$ is defined by \begin{equation}
   \begin{vmatrix}
\delta X^{q^{2t}} & CX^{q^k}\\
\delta Y^{q^{2t}} & CY^{q^k}
\end{vmatrix}=\delta C(X^{q^{2t}}Y^{q^k} - X^{q^k}Y^{q^{2t}})=\delta C X^{q^{2t}}Y^{q^{2t}}\prod_{\xi \in \mathbb{F}_{q^{k-2t}} \setminus \{0\}}  (Y - \xi X)^{q^{2t}}=0,
 \end{equation} 
so points at infinity of $\mathcal{C}$ are of type $P=(1:\xi:0),$ $\xi \in \mathbb{F}_{q^{k-2t}},$ or $P=(0:1:0).$ \\
Let consider the projectivity $\Psi:(x:y:z) \mapsto (x:y-\xi x: z)$  that maps $(1:\xi:0)$ into $(1:0:0).$
An affine equation of $\Psi(\mathcal{C})$ (dehomogenizing with respect to $X$) is given by
\begin{equation}  \label{eq f punti infinito}
f(Y,Z):=\begin{vmatrix}
1 & Z^{q^{2t}-1}+\delta & Z^{q^k-q^{2t}}+\dots +C\\
(Y+\xi)^{{q^t}} & F^*(Y+\xi,Z) & G^*(Y+\xi, Z) \\
\lambda^{{q^t}} & F(\lambda)Z^{q^{2t}-q^t} & G(\lambda)Z^{q^k-q^t}
\end{vmatrix}=0,
\end{equation}
where
\begin{eqnarray}
    F^*(Y+\xi,Z)&=&YZ^{q^{2t}-1}+\delta Y^{q^{2t}}+\xi Z^{q^{2t}-1} + \delta\xi^{q^{2t}};  \\
    G^*(Y+\xi, Z)&=&Y^{q^{2t}}Z^{q^k-q^{2t}}+\dots +CY^{q^k}+\xi^{q^{2t}}Z^{q^k-q^{2t}}+\dots + C\xi^{q^k}.
\end{eqnarray}
It is not hard to see that
\begin{equation*}
\begin{vmatrix}
1 & Z^{q^{2t}-1}+\delta & Z^{q^k-q^{2t}}+\dots +C\\
(Y+\xi)^{{q^t}} & F^*(Y+\xi,Z) & G^*(Y+\xi, Z) \\
\lambda^{{q^t}} & F(\lambda)Z^{q^{2t}-q^t} & G(\lambda)Z^{q^k-q^t}
\end{vmatrix}=B_{q^{2t}-q^t}+B_{q^{2t}-1}+B_{q^{2t}} + L(Y,Z),
\end{equation*}
where
\begin{eqnarray}
 B_{q^{2t}-q^t}&=&-CF(\lambda)(\xi^{q^{k-t}}-\xi)^{q^t}Z^{q^{2t}-q^t},\\
 B_{q^{2t}-1}&=&C\lambda^{q^t}(\xi^{q^k}-\xi)Z^{q^{2t}-1}, \\
 B_{q^{2t}}&=&C[ F(\lambda)Y^{q^t}Z^{q^{2t}-q^t} -\lambda^{q^t}(YZ^{q^{2t}-1}+\delta Y^{q^{2t}})],
\end{eqnarray}
and
\begin{equation}
    L(Y,Z)=\sum_{i,j}\alpha^{(i,j)}Y^iZ^j, \ \ \ \ \ \ \ \ \ i+j \geq q^k-q^{k-1}=q^{k-1}(q-1)\geq q^{2t}(q-1).
\end{equation}
\begin{prop}  \label{prop molteplicità punti singolari all'inf}
    Let $P_\xi=(1: \xi : 0), \ \xi \in \mathbb{F}_{q^{k-2t}}.$ Then
    $m_{P_{\xi}}(\mathcal{C})=q^{2t}-q^t$ or  $m_{P_{\xi}}(\mathcal{C})=q^{2t}.$ \\
    Also $m_{P_{\xi}}(\mathcal{C})=q^{2t}$ \  if and only if  \   $\xi \in \mathbb{F}_{q^{\gcd(k,t)}}.$
\end{prop}
\begin{proof}
   The statement follows immediately from 
   \begin{eqnarray*}
       &&\mathbb{F}_{q^{k-2t}}\cap \mathbb{F}_{q^{k-t}}=\mathbb{F}_{q^{\gcd(k-2t,k-t)}}, \\
       &&\mathbb{F}_{q^{k-t}}\cap \mathbb{F}_{q^{k}}=\mathbb{F}_{q^{\gcd(k-t,k)}}, \\
       &&\gcd(k-2t,k-t)=\gcd(k-t,t)=\gcd(k,t)=\gcd(k,k-t).
   \end{eqnarray*}
     \end{proof}

\begin{prop} \label{prop stime ipmax punti all'infinito}
    Let $P=(1:\xi:0),$ where $\xi \in \mathbb{F}_{q^{\gcd(k,t)}},$ or $P=(0:1:0).$ Then $I_{P,max}(\mathcal{C})\leq \frac{q^{4t}}{4}.$
\end{prop}
\begin{proof}
    \begin{itemize}
        \item $P=(1:\xi:0),$ with $\xi \in \mathbb{F}_{q^{\gcd(k,t)}}.$ \\ The statement follows from Lemma \ref{lemma easy} and Proposition \ref{prop molteplicità punti singolari all'inf}.
        \item $P=(0:1:0).$ \\
    Since $\mathcal{C}$ is fixed by the projectivity $(x:y:z) \mapsto (y:x:z),$ we get that $I_{(0:1:0),max}=I_{(1:0:0),max} \leq \frac{q^{4t}}{4}.$
        \end{itemize}
\end{proof}
 \subsection{Study of the intersection multiplicity of branches at the singular points $P_{\xi}=(1:\xi:0),$ with $\xi \notin \mathbb{F}_{q^{\gcd(k,t)}}$}
In this subsection we consider singular points of type $P_{\xi}=(1:\xi:0),$ where $\xi \notin \mathbb{F}_{q^{\gcd(k,t)}}.$ In particular we study branches centred at those points an we determine upper bounds on the multiplicity of intersection of putative components of $\mathcal{C}$ (and therefore $\mathcal{A}$).

\begin{lem}  \label{lem un solo ramo}
    Let $P_\xi=(1:\xi:0),$ with $\xi \notin \mathbb{F}_{q^{\gcd(k,t)}}.$ Then there is a unique branch centered at $P_\xi.$ Thus, the multiplicity of two putative components of $\mathcal{C}$ (and therefore $\mathcal{A}$) in $P_\xi$ is $0.$
\end{lem}
\begin{proof}
    Consider the local quadratic transformation $f^{(1)}(Y,Z)=f(Y,YZ)/Y^{q^{2t}-q^t},$ where $f(Y,Z)$ is defined as in (\ref{eq f punti infinito}). We have that \begin{eqnarray*}
    f^{(1)}(Y,Z)&=&\alpha^{(0,q^{2t}-q^t)}Z^{q^{2t}-q^t}+\alpha^{(0,q^{2t}-1)}Y^{q^t-1}Z^{q^{2t}-1}+\alpha^{(q^t,q^{2t}-q^t)}Y^{q^t}Z^{q^{2t}-q^t} \\&+&\alpha^{(1,q^{2t}-1)}Y^{q^t}Z^{q^{2t}-1}+\alpha^{(q^{2t},0)}Y^{q^{t}}+L^{(1)}(Y,Z),
\end{eqnarray*}
where
\begin{eqnarray}
    \alpha^{(0,q^{2t}-q^t)}&=&-CF(\lambda)(\xi^{q^{k-t}}-\xi)^{q^t};\\
    \alpha^{(0,q^{2t}-1)}&=&C\lambda^{q^t}(\xi^{q^k}-\xi);\\
    \alpha^{(1,q^{2t}-1)}&=&-C\lambda^{q^t};\\
    \alpha^{(q^{2t},0)}&=&-C\lambda^{q^t}\delta;\\
    \alpha^{(q^t,q^{2t}-q^t)}&=&CF(\lambda);\\
    L^{(1)}(Y,Z)&=&\sum_{i,j}\alpha^{(i,j)}Y^{i+j-q^{2t}+q^{t}}Z^j.
\end{eqnarray}
Since $f^{(1)}(0,Z)=\alpha^{(0,q^{2t}-q^t)}Z^{q^{2t}-q^t},$ the branches centered at $(0,0)$ belonging to $f=0$ correspond to the branches centered at $(0,0)$ belonging to $f^{(1)}=0.$ Now the origin is a non-ordinary $q^t$-point of the curve defined by $f^{(1)}$ and the line $Z=0$ is not the tangent line at the origin; so we can apply the local quadratic transformation $f^{(2)}(Y,Z)=f^{(1)}(YZ,Z)/Z^{q^t}.$ After applying the same trasformation $q^t-2$ times, we obtain
\begin{eqnarray*}
    f^{(q^t-1)}(Y,Z)&=&\alpha^{(0,q^{2t}-q^t)}Z^{q^t}+\alpha^{(0,q^{2t}-1)}Y^{q^t-1}Z^{q^{2t}-q^t+1}+\alpha^{(q^t,q^{2t}-q^t)}Y^{q^t}Z^{q^{2t}-q^t} \\&+&\alpha^{(1,q^{2t}-1)}Y^{q^t}Z^{q^{2t}-1}+\alpha^{(q^{2t},0)}Y^{q^{t}}+L^{(q^t-1)}(Y,Z),
\end{eqnarray*}
where $L^{(q^t-1)}(Y,Z)=\sum\limits_{i,j}\alpha^{(i,j)}Y^{i+j-q^{2t}+q^{t}}Z^{j+(q^t-2)(i+j-q^{2t})}.$ \\Note that, for each $2\leq \ell \leq q^t-1,$ $L^{(\ell)}(Y,Z)=\sum\limits_{i,j}\alpha^{(i,j)}Y^{i+j-q^{2t}+q^{t}}Z^{j+(\ell -1)(i+j-q^{2t})}.$ \\
Also, $j+(\ell -1)(i+j-q^{2t}) \geq i+j-q^{2t} >0,$ so $L^{(\ell)}(Y,0)=0$ and $f^{(\ell)}(Y,0)=\alpha^{(q^{2t},0)}Y^{q^{t}}.$
Therefore, the branches centered at $(0,0)$ belonging to $f=0$ are in one-to-one correspondence with the branches centered at $(0,0)$ belonging to $f^{(\ell)}=0,$ for each $2\leq \ell \leq q^t-1.$\\
Now consider the transformation $(Y,Z) \mapsto (Y-\beta_1Z,Z),$ where $\alpha^{(q^{2t},0)}\beta_1^{q^t}+\alpha^{(0,q^{2t}-q^t)}=0,$ i.e.
\begin{equation} \label{eq beta1}
    \beta_1^{q^t}=-\frac{F(\lambda)(\xi^{q^{k-t}}-\xi)^{q^t}}{\lambda^{q^t} \delta}.
\end{equation}
We get \begin{eqnarray*}
    f^{(q^t)}(Y,Z)&=&f^{(q^t-1)}(Y+\beta_1Z,Z)=\alpha^{(0,q^{2t}-1)}\sum_{\mu=0}^{q^t-1}\beta_1^{q^t-1-\mu}Y^\mu Z^{q^{2t}-\mu}\\&+&\alpha^{(q^t,q^{2t}-q^t)}(Y^{q^t}Z^{q^{2t}-q^t} +\beta_1^{q^t}Z^{q^{2t}})\\&+&\alpha^{(1,q^{2t}-1)}(Y^{q^t}Z^{q^{2t}-1}+\beta_1^{q^t}Z^{q^{2t}+q^t-1})\\&+&\alpha^{(q^{2t},0)}Y^{q^{t}}+L^{(q^t)}(Y,Z),
\end{eqnarray*}
where $L^{(q^t)}(Y,Z)=\sum\limits_{i+j\geq q^{2t}(q-1)}\alpha^{(i,j)}\sum\limits_{\nu=0}^{i+j-q^{2t}+q^{t}}\binom{i+j-q^{2t}+q^t}{\nu}\beta_1^{i+j-q^{2t}+q^{t}-\nu}Y^{\nu}Z^{j+(q^t -1)(i+j-q^{2t})+q^t-\nu}.$ \\
Since the origin is still a non-ordinary $q^t$-point of  $f^{(q^t)}=0$ and the line $Z=0$ is not the tangent line at the origin, we can apply the local quadratic transformation $f^{(q^t+1)}(Y,Z)=f^{(q^t)}(YZ,Z)/Z^{q^t}.$ After applying the same trasformation $q^t-1$ times, we obtain 
\begin{eqnarray*}
    f^{(2q^t-1)}(Y,Z)&=&\alpha^{(0,q^{2t}-1)}\sum_{\mu=0}^{q^t-1}(-1)^\mu\beta_1^{q^t-1-\mu}Y^\mu Z^{q^{t}+(q^t-2)\mu}+\alpha^{(q^t,q^{2t}-q^t)}(Y^{q^t}Z^{q^{2t}-q^t} +\beta_1^{q^t}Z^{q^{t}})\\&+&\alpha^{(1,q^{2t}-1)}(Y^{q^t}Z^{q^{2t}-1}+\beta_1^{q^t}Z^{2q^t-1})+\alpha^{(q^{2t},0)}Y^{q^{t}}+L^{(2q^t-1)}(Y,Z),
\end{eqnarray*}
where $L^{(2q^t-1)}(Y,Z)=\!\!\!\!\sum\limits_{i+j\geq q^{2t}(q-1)}\alpha^{(i,j)}\sum\limits_{\nu=0}^{i+j-q^{2t}+q^{t}}\binom{i+j-q^{2t}+q^t}{\nu}\beta_1^{i+j-q^{2t}+q^{t}-\nu}Y^{\nu}Z^{j+(q^t -1)(i+j-q^{2t})+(q^t-2)(\nu-q^t)}.$ \\
Note that, for each $1\leq \ell \leq q^t-2,$ $$j+(q^t -1)(i+j-q^{2t})+\ell(\nu-q^t)\geq j+\ell(i+j-q^{2t}-q^t+\nu)\geq i+j-q^{2t}(1+\frac{1}{q^t})>0,$$
so $L^{(\ell)}(Y,0)=0$ and $f^{(\ell)}(Y,0)=\alpha^{(q^{2t},0)}Y^{q^{t}}.$
Therefore the branches centered at $(0,0)$ for $f=0$ are in one-to-one correspondence with the branches centered at $(0,0)$ for $f^{(2q^t-1)}=0.$
Consider the map $(Y,Z)\mapsto (Y-\beta_2Z,Z),$ with $\alpha^{(0,q^{2t}-1)}\beta_1^{q^t-1}+\alpha^{(q^t,q^{2t}-q^t)}\beta_1^{q^t}+\alpha^{(q^{2t},0)}\beta_2^{q^t}=0,$ i.e. 
\begin{eqnarray*}
    \beta_2^{q^{2t}}&=&\dfrac{\beta_1^{q^t(q^t-1)}[\lambda^{q^t}(\xi^{q^k}-\xi)+F(\lambda)\beta_1]^{q^t}}{\lambda^{q^{2t}}\delta^{q^t}}\\
    &=&\dfrac{F(\lambda)^{q^t-1}(\xi^{q^{k-t}}-\xi)^{q^t(q^t-1)}[\lambda^{q^t(q^t+1)}\delta (\xi^{q^k}-\xi)^{q^t}-F(\lambda)^{q^t+1}(\xi^{q^{k-t}}-\xi)^{q^t}]}{\lambda^{2q^{2t}}\delta^{2q^t}}.
\end{eqnarray*}
We get 
\begin{eqnarray*}
    f^{(2q^t)}(Y,Z)\!\!\!&=&\!\!\!\alpha^{(0,q^{2t}-1)}\!\sum_{\mu=1}^{q^t-1}\!\sum_{\Tilde{\mu}=0}^{\mu}(-1)^\mu\binom{\mu}{\Tilde{\mu}}\beta_1^{q^t-1-\mu}\beta_2^{\mu-\Tilde{\mu}}Y^{\Tilde{\mu}} Z^{q^{t}+(q^t-1)\mu-\Tilde{\mu}}+\alpha^{(q^t,q^{2t}-q^t)}(Y^{q^t}Z^{q^{2t}-q^t}+\beta_2^{q^t}Z^{q^{2t}})\\\!\!\!&+&\!\!\!\alpha^{(1,q^{2t}-1)}(Y^{q^t}Z^{q^{2t}-1}+\beta_2^{q^t}Z^{q^{2t}+q^t-1}+\beta_1^{q^t}Z^{2q^t-1})+\alpha^{(q^{2t},0)}Y^{q^{t}}+L^{(2q^t)}(Y,Z),
\end{eqnarray*}
where $L^{(2q^t)}(Y,Z)=\sum\limits_{i+j\geq q^{2t}(q-1)}\sum\limits_{\nu=0}^{i+j-q^{2t}+q^{t}}\sum\limits_{\Tilde{\nu}=0}^{\nu}\Tilde{\alpha}_{\nu\Tilde{\nu}}^{(i,j)}Y^{\Tilde{\nu}}Z^{j+(q^t -1)(i+j-q^{2t})+(q^t-2)(\nu-q^t)+\nu-\Tilde{\nu}},$ \\
for suitable $\Tilde{\alpha}_{\nu\Tilde{\nu}}^{(i,j)} \in \mathbb{F}_{q^n}.$ Let apply again the quadratic transformation $f^{(2q^t+1)}(Y,Z)=f^{(2q^t)}(YZ,Z)/Z^{q^t}.$ It follows that
\begin{eqnarray*}
    f^{(2q^t+1)}(Y,Z)\!\!\!&=&\!\!\!\alpha^{(0,q^{2t}-1)}\sum_{\mu=1}^{q^t-1}\sum_{\Tilde{\mu}=0}^{\mu}(-1)^\mu\binom{\mu}{\Tilde{\mu}}\beta_1^{q^t-1-\mu}\beta_2^{\mu-\Tilde{\mu}}Y^{\Tilde{\mu}} Z^{(q^t-1)\mu}+\alpha^{(q^t,q^{2t}-q^t)}(Y^{q^t}Z^{q^{2t}-q^t}+\beta_2^{q^t}Z^{q^{2t}-q^t})\\\!\!\!&+&\!\!\!\alpha^{(1,q^{2t}-1)}(Y^{q^t}Z^{q^{2t}-1}+\beta_2^{q^t}Z^{q^{2t}-1}+\beta_1^{q^t}Z^{q^t-1})+\alpha^{(q^{2t},0)}Y^{q^{t}}+L^{(2q^t+1)}(Y,Z),
\end{eqnarray*}
where $L^{(2q^t+1)}(Y,Z)=\sum\limits_{i+j\geq q^{2t}(q-1)}\sum\limits_{\nu=0}^{i+j-q^{2t}+q^{t}}\sum\limits_{\Tilde{\nu}=0}^{\nu}\Tilde{\alpha}_{\nu\Tilde{\nu}}^{(i,j)}Y^{\Tilde{\nu}}Z^{j+(q^t -1)(i+j-q^{2t}+\nu-q^t)}.$ \\
Thus the smallest homogeneous part of $f^{(2q^t+1)}(Y,Z)$ is given by $(\alpha^{(1,q^{2t}-1)}\beta_1^{q^t}-\alpha^{(0,q^{2t}-1)}\beta_1^{q^t-2}\beta_2)Z^{q^t-1},$\\

where 
\begin{eqnarray*}
(\alpha^{(1,q^{2t}-1)}\beta_1^{q^t}-\alpha^{(0,q^{2t}-1)}\beta_1^{q^t-2}\beta_2)^{q^{2t}}&=&-C^{q^{2t}}\lambda^{q^{3t}}\beta_1^{q^{2t}(q^t-2)}[\beta_1^{2q^{2t}}+(\xi^{q^k}-\xi)^{q^t}\beta_2^{q^{2t}}].   
\end{eqnarray*}
Therefore from  (\ref{eq beta1}) it follows that $m_{(0,0)}(f^{(2q^t+1)})=q^t-1$ if and only \begin{equation*}
    \beta_1^{2q^{2t}}+(\xi^{q^k}-\xi)^{q^t}\beta_2^{q^{2t}}\neq 0,
\end{equation*}
i.e., after some computations,
\begin{equation}
    \dfrac{F(\lambda)^{q^{t}-1}}{\lambda^{2q^{2t}}\delta^{2q^t}}(\xi^{q^{k-t}}-\xi)^{q^{2t}-q^t}[(\xi-\xi^{q^{t}})^{q^{k+t}}(\xi^{q^{k-t}}-\xi)^{q^t}F(\lambda)^{q^t+1}+\delta(\xi^{q^k}-\xi)^{q^t(q^t+1)}\lambda^{q^t(q^t+1)}]\neq 0
\end{equation}
(recall Equation \ref{lambda scelto opportunamente}).

Also, the homogeneous part of degree $q^t$ in $f^{(2q^t+1)}(Y,Z)$ is not zero and contains the term $\alpha^{(q^{2t},0)}Y^{q^{t}}.$ By Lemma \ref{lemma easy}(i) it follows that $I_{P_\xi,max}(\mathcal{C})=0,$ i.e. there is a unique branch of $\mathcal{C}$ centered at $P_\xi.$
\end{proof}

\section{Main result}
\subsection{Case $\min\deg_q(G(X))=2t$}
Thanks to the results of the previous section, we now are able to prove the following theorem. 
\begin{thm} \label{th main caso 2t}
    Let $q,t,k$ be intergers, with $q>2,$ $t>0$ and $k>2t.$ If $(t,q) \notin \{(1,3);(1,4);(1,5);(2,3)\}$ and  $n\geq k+t+1,$ then $\mathcal{A}$ has an absolutely irreducible component defined over $\mathbb{F}_{q^n}$ and not contained in the curve defined by $F_{(x,x^q,x^{q^2})}(X,Y,\lambda)=0.$ 
\end{thm}
\begin{proof} Since $\mathcal{C}$ has a finite number of singularities, the curve defined by $F_{(x,x^q,x^{q^2})}(X,Y,\lambda)=0$ is not a component of $\mathcal{A}.$
Moreover, any upper bound on $I_{P,max}(\mathcal{C})$ is also an upper bound for $I_{P,max}(\mathcal{A})$, since the curve $\mathcal{A}$ is a component of the curve $\mathcal{C}$. \\
We want to prove that $\displaystyle\sum_{P\in Sing(\mathcal{A})}I_{P,max}(\mathcal{A}) < \frac{2}{9}\deg^{2}(\mathcal{A})=\frac{2}{9}(q^k+q^{2t}-(q^2+q+1))^2;$  the statement will then follow by Criterion \ref{criterio due noni}.\\
  Let 
	\begin{eqnarray*}
	\Omega&:=&\{(1:\xi:0):\xi \in \mathbb{F}_{q^{k-2t}}\}\cup \{(0:1:0)\},\\
       \Pi&:=&\{(1:\xi:0):\xi \in \mathbb{F}_{q^{\gcd(k,t)}}\},\\
	\Theta &:=& \{(\overline{x},\overline{y}) : \overline{x}^{q^t}G(\lambda)-\lambda^{q^t}G(\overline{x})=\overline{y}^{q^t}G(\lambda)-\lambda^{q^t}G(\overline{y})=0 \},\\
	\Sigma &:=& \{(\overline{x},\overline{y}) \in \Theta : F(\overline{x})\lambda^{q^t}-F(\lambda)\overline{x}^{q^t}=F(\overline{y})\lambda^{q^t}-F(\lambda)\overline{y}^{q^t}=0\}.
	\end{eqnarray*}  
 We have that $Sing(\mathcal{C})=\Omega \cup \Theta,$ $\#\Omega=q^{k-2t}+1,$ $\#\Pi=q^{\gcd(k,t)}+1,$ $\#\Theta\leq q^{2(k-t)},$ $\#\Sigma\leq q^{\min\{\max\{ 2(k-2t),2t\},4t\}}$ (see Remark \ref{rem numero di singolarità affini}) and  by Proposition \ref{prop stime ipmax caso affine}, Proposition \ref{prop stime ipmax punti all'infinito}, and Lemma \ref{lem un solo ramo}	
	\begin{eqnarray*}
	I_{P,\max}(\mathcal{A})\leq 
	\begin{cases}
	    0 & P \in \Omega \setminus \Pi;\\
	    \frac{q^{4t}}{4} & P \in \Pi;\\
	    q^t  & \Theta \setminus \Sigma; \\
    \frac{(q^t+1)^2}{4} & P \in \Sigma.
	\end{cases} 
	\end{eqnarray*}
Observe that $Sing(\mathcal{C})$ is a finite set, so the curve defined by  $F_{(x,x^q,x^{q^2})}(X,Y,\lambda)=0$ is not a component of $\mathcal{A}.$ 
We distinguish the following cases.
\begin{itemize}
    \item Case $k > 3t$ \\
It follows that
\begin{eqnarray*}  
    \displaystyle\sum_{P\in Sing(\mathcal{A})}I_{P,max}(\mathcal{A}) &\leq& (q^{2(k-t)}-q^{4t})q^t+q^{4t}\cdot\frac{(q^t+1)^{2}}{4}+(q^{\gcd(k,t)}+1)\frac{q^{4t}}{4} \\
&\leq&q^{2k-t}-q^{5t}+\frac{q^{6t}}{4}+\frac{q^{5t}}{2}+\frac{q^{4t}}{4}+\frac{q^{5t}}{4}+\frac{q^{4t}}{4}=q^{2k-t}+\frac{q^{6t}}{4}-\frac{q^{5t}}{4}+\frac{q^{4t}}{2}\\
    &\leq& q^{2k-t}+\frac{q^{6t}}{4}\leq
    \frac{5}{4}q^{\max\{2k-t,6t\}} \leq \frac{2}{9}q^{2k} < \frac{2}{9}(q^k+q^{2t}-(q^2+q+1))^2,
\end{eqnarray*}
for $t\geq 2$ and $q \geq 3.$ \\
Finally, for $t=1,$  the inequality \begin{equation}
    (q^{2(k-t)}-q^{4t})q^t+q^{4t}\cdot\frac{(q^t+1)^{2}}{4}+(q^{\gcd(k,t)}+1)\frac{q^{4t}}{4} < \frac{2}{9}(q^k+q^{2t}-(q^2+q+1))^2
\end{equation} is satisfied when even
$$\frac{2}{9} - \frac{1}{q} -\frac{4(q+1)}{9q^{k}}-\frac{1}{4q^2} > 0,$$ that holds for $q\geq 5.$
\item Case $k = 3t$ \\
We have 
 \begin{eqnarray} \label{eq caso k eq 3t}
 \displaystyle\sum_{P\in Sing(\mathcal{A})}I_{P,max}(\mathcal{A}) &\leq& (q^{4t}-q^{2t})q^t+q^{2t}\cdot\frac{(q^t+1)^{2}}{4}+(q^{t}+1)\frac{q^{4t}}{4} \\
    &\leq& \frac{5}{4}q^{5t}+\frac{q^{4t}}{2}-\frac{q^{3t}}{2}+\frac{q^{2t}}{4} \leq \frac{5}{4}q^{5t} + \frac{q^{4t}}{2} \nonumber \\ &\leq &\frac{2}{9}q^{6t} 
    < \frac{2}{9}(q^{3t}+q^{2t}-(q^2+q+1))^2, \nonumber
 \end{eqnarray}
 for $q^t\geq 6,$ so $q>2$ and $t\geq 2.$ 
 \\
 For $t=1,$ the inequality $\displaystyle\sum_{P\in Sing(\mathcal{A})}I_{P,max}(\mathcal{A}) < \frac{2}{9}q^{2k}$ is satisfied when 
 $$\frac{2}{9} - \frac{5}{4q} -\frac{4(q+1)}{9q^{3}}-\frac{1}{2q^2} > 0, $$ that holds for $q\geq 7.$

 \item Case $2t < k < 3t$ \\
 It results 
 \begin{eqnarray}
 \displaystyle\sum_{P\in Sing(\mathcal{A})}I_{P,max}(\mathcal{A}) \leq (q^{2(k-t)}-q^{2t})q^t+q^{2t}\cdot\frac{(q^t+1)^{2}}{4}+(q^{s}+1)\frac{q^{4t}}{4},
 \end{eqnarray}
 where $s:=\gcd(k,t),$ $s \leq t/2,$ $k=\Tilde{k}s,$ $t=\Tilde{t}s,$ $\gcd(\Tilde{k},\Tilde{t})=1.$ \\
 Note that $k > 2t $ implies $\Tilde{k} \geq 2 \Tilde{t}+1,$ so $s+4t=(1+4\Tilde{t})s < (1+2\Tilde{t})2s \leq 2\Tilde{k}s=2k.$ \\
 It follows that 
 \begin{eqnarray*}
 \displaystyle\sum_{P\in Sing(\mathcal{A})}I_{P,max}(\mathcal{A}) &\leq& 
     q^{2k-t}+\frac{q^{4t}}{2}-\frac{q^{3t}}{2}+\frac{q^{2t}}{4}+\frac{q^{4t+s}}{4}\\
     &\leq& q^{2k-t}+\frac{q^{4t+s}}{2} \leq \left\{
	\begin{array}{l}
	\frac{3}{2}q^{2k-1} \ \ \  \ \ \ \  \ \ \textnormal{for} \ t\geq 1\\ 
		(1+\frac{q}{2})q^{2k-2} \ \ \textnormal{for} \ t\geq 2  \\
  (1+\frac{q^2}{2})q^{2k-3} \ \ \textnormal{for} \ t\geq 3
	\end{array}
	\right. \\ &\leq&  \frac{2}{9}q^{2k} < \frac{2}{9}(q^{k}+q^{2t}-(q^2+q+1))^2,
 \end{eqnarray*}
 for $t \geq 3$ and $q\geq 3,$ $t\geq 2$ and $q\geq 4,$  or $t=1$ and $q\geq 7.$ \\
\end{itemize}
\end{proof}
\subsection{Case $\min\deg_q(G(X))=t/2$}
The procedure is exactly the same as for $\min\deg_q(G(X))=2t.$ \\
Consider  \begin{equation} \label{eq G caso t mezzi}
    G:=X^{q^{t/2}} + \dots + CX^{q^k} \ \ \ \ \ C\neq 0, \ \ 0 < 2t < k < n, 
\end{equation}
and $F,\lambda, \mathcal{C},\mathcal{A}$
as in Equations \ref{eq f e g}, \ref{lambda scelto opportunamente}, \ref{eq curva C}, \ref{eq curva A}.
Actually, all the results in Section $6$  apply here. In particular, by repeating all the computations  in Section $6$ we get the analogous of Propositions \ref{prop molteplicità punti singolari all'inf}, \ref{prop stime ipmax punti all'infinito} and Lemma \ref{lem un solo ramo}.
\begin{prop} \label{prop caso t mezzi ipmax singolarità infinito}
    Let $n\geq k+t+1$ and $F,G,\lambda,\mathcal{C},\mathcal{A}$ as in Equations \ref{eq f e g}, \ref{lambda scelto opportunamente}, \ref{eq curva C}, \ref{eq curva A}, \ref{eq G caso t mezzi}.  Then \begin{itemize}
        \item $Sing(\mathcal{C})\cap \{Z=0\}=\{(1:\xi:0):\xi \in \mathbb{F}_{q^{k-2t}}\}\cup \{(0:1:0)\},$ 
        \item $m_{(1:\xi:0)}(\mathcal{C})=q^{2t}-q^t$ or  $m_{(1:\xi:0)}(\mathcal{C})=q^{2t},$ 
    \item $m_{(1:\xi:0)}(\mathcal{C})=q^{2t}$ \  if and only if  \   $\xi \in \mathbb{F}_{q^{\gcd(k,t)}},$
    \item $I_{(1:\xi:0),\max}(\mathcal{C})\leq 
	\begin{cases}
	    \frac{q^{4t}}{4} & \textnormal{if} \ \ \xi \in \mathbb{F}_{q^{\gcd(k,t)}};\\
	    0 & \textnormal{if} \ \  \xi \in \mathbb{F}_{q^{k-2t}} \setminus  \mathbb{F}_{q^{\gcd(k,t)}}.
	\end{cases}$
    \end{itemize} 
\end{prop}
For affine singularities of $\mathcal{C}$, we proceed exactly as in Section 5. \\
A point $P=(\overline{x},\overline{y}) \in \mathcal{C}$ is singular if and only if
\begin{equation*}
\begin{vmatrix}
\overline{x}^{{q^t}} & G(\overline{x})\\
\lambda^{{q^t}} &  G(\lambda)
\end{vmatrix}=\begin{vmatrix}
\overline{y}^{{q^t}} & G(\overline{y})\\
\lambda^{{q^t}} &  G(\lambda)
\end{vmatrix}=0.
\end{equation*}
By direct computations one can observe that
\begin{equation*}
    \begin{vmatrix}
(X+\overline{x})^{{q^t}} & F(X+\overline{x}) & G(X+\overline{x})\\
(Y+\overline{y})^{{q^t}} & F(Y+\overline{y}) & G(Y+\overline{y}) \\
\lambda^{{q^t}} & F(\lambda) & G(\lambda)
\end{vmatrix}=H_{q^{t/2}}+H_{q^{t/2}+1}+\dots,
\end{equation*}
where
\begin{eqnarray}
  \label{t mezzi cono}H_{q^{t/2}}&=&\begin{vmatrix}  
\overline{y}^{q^t} & F(\overline{y})\\
\lambda^{q^t} &  F(\lambda)
\end{vmatrix}X^{q^{t/2}}-\begin{vmatrix}
\overline{x}^{q^t} & F(\overline{x})\\
\lambda^{q^t} &  F(\lambda)
\end{vmatrix}Y^{q^{t/2}}; \\
  H_{q^{t/2}+1}&=&\lambda^{q^t}(X^{q^{t/2}}Y-XY^{q^{t/2}}).
\end{eqnarray}

Thus $I_{P,max}(\mathcal{C}) \leq q^{t/2}$ or $I_{P,max}(\mathcal{C})\leq \frac{(q^{t/2}+1)^2}{4}$ by Lemma \ref{lemma easy}. \\
Also, observe that
\begin{equation*}
\begin{vmatrix}
\overline{x}^{{q^t}} & G(\overline{x})\\
\lambda^{{q^t}} &  G(\lambda)
\end{vmatrix}=[\overline{x}^{q^{t/2}}G(\lambda)^{1/q^{t/2}}-\lambda^{q^{t/2}}(\overline{x}+\dots +C^{1/q^{t/2}}\overline{x}^{q^{k-{t/2}}})]^{q^{t/2}}=g(\overline{x})^{q^{t/2}},
\end{equation*}
where $g$ is a separable polynomial of degree $q^{k-{t/2}}.$  Thus the number of affine singular points of $\mathcal{C}$ is at most $q^{2k-t}.$ 
Finally, if $m_P(\mathcal{C})=q^{t/2}+1,$ that is \begin{equation}
  \begin{vmatrix}
\overline{y}^{q^t} & F(\overline{y})\\
\lambda^{q^t} &  F(\lambda)
\end{vmatrix}=\begin{vmatrix}
\overline{x}^{q^t} & F(\overline{x})\\
\lambda^{q^t} &  F(\lambda)
\end{vmatrix}=0, 
\end{equation} we obtain  \begin{equation} \label{eq F x segnato lambda caso t mezzi}
    F(\overline{x})\lambda^{q^t}-F(\lambda)\overline{x}^{q^t}=\delta \lambda^{q^t} \overline{x}^{q^{2t}}-(\lambda +\delta \lambda^{q^{2t}})\overline{x}^{q^t} +\lambda^{q^t}\overline{x}=0.
\end{equation}
By combining 
last equation with $g(\overline{x})=\overline{x}^{q^{t/2}}G(\lambda)^{1/q^{t/2}}-\lambda^{q^{t/2}}(\overline{x}+\dots +C^{1/q^{t/2}}\overline{x}^{q^{k-{t/2}}})=0,$ we get that $\overline{x}$ (or $\overline{y}$ equivalently) must be a root of a $q$-polynomial 
\begin{equation}
    h(X)^{q^{t/2}}:=[\alpha_0 X + \dots + \alpha_{k-2t} X^{q^{k-t}}+\Tilde{\alpha}X^{q^{3t/2}}]^{q^{t/2}},
\end{equation}
for suitable $\alpha_0,\dots,\alpha_{k-2t}, \Tilde{\alpha} \in \mathbb{F}_{q^n},$ where $\deg(h)\leq q^{\max\{k-t,3t/2\}}.$

Now we are able to prove the analogous of Theorem \ref{th main caso 2t} for $\min\deg_q(G)=t/2.$
\begin{thm} \label{thm A componente caso t mezzi} Let $q,t,k$ be integers, with $q>2,$ $t>0$ even and $k>2t.$ Let $n\geq k+t+1$ and $F,G,\lambda,\mathcal{C},\mathcal{A}$ as in Equations \ref{eq f e g}, \ref{lambda scelto opportunamente}, \ref{eq curva C}, \ref{eq curva A}, \ref{eq G caso t mezzi}. If $(t,q) \notin \{ (2,3);(2,4);(2,5),(4,3)\},$ then $\mathcal{A}$ has an absolutely irreducible component defined over $\mathbb{F}_{q^n}$ and not contained in the curve defined by $F_{(x,x^q,x^{q^2})}(X,Y,\lambda)=0.$
\end{thm}
\begin{proof}
Since $ |Sing(\mathcal{C})| < \infty,$ the curve defined by $F_{(x,x^q,x^{q^2})}(X,Y,\lambda)=0$ is not a component of $\mathcal{A};$ thus by Criterion \ref{criterio due noni} it is enough  to prove that $\displaystyle\sum_{P\in Sing(\mathcal{A})}I_{P,max}(\mathcal{A}) < \frac{2}{9}\deg^{2}(\mathcal{A})=\frac{2}{9}(q^k+q^{2t}-(q^2+q+1))^2.$ \\
  Let 
	\begin{eqnarray*}
	\Omega&:=&\{(1:\xi:0):\xi \in \mathbb{F}_{q^{k-2t}}\}\cup \{(0:1:0)\},\\
       \Pi&:=&\{(1:\xi:0):\xi \in \mathbb{F}_{q^{\gcd(k,t)}}\},\\
	\Theta &:=& \{(\overline{x},\overline{y}) : \overline{x}^{q^t}G(\lambda)-\lambda^{q^t}G(\overline{x})=\overline{y}^{q^t}G(\lambda)-\lambda^{q^t}G(\overline{y})=0 \},\\
	\Sigma &:=& \{(\overline{x},\overline{y}) \in \Theta : F(\overline{x})\lambda^{q^t}-F(\lambda)\overline{x}^{q^t}=F(\overline{y})\lambda^{q^t}-F(\lambda)\overline{y}^{q^t}=0\}.
	\end{eqnarray*}  
 We have that $Sing(\mathcal{C})=\Omega \cup \Theta,$ $\#\Omega=q^{k-2t}+1,$ $\#\Pi=q^{\gcd(k,t)}+1,$ $\#\Theta\leq q^{2k-t},$ $\#\Sigma\leq q^{\min\{\max\{ 2(k-t),3t\},4t\}},$ and  by Equation \ref{t mezzi cono}  and Proposition \ref{prop caso t mezzi ipmax singolarità infinito}	
	\begin{eqnarray*}
	I_{P,\max}(\mathcal{A})\leq 
	\begin{cases}
	    0 & P \in \Omega \setminus \Pi;\\
	    \frac{q^{4t}}{4} & P \in \Pi;\\
	    q^{t/2}  & \Theta \setminus \Sigma; \\
    \frac{(q^{t/2}+1)^2}{4} & P \in \Sigma.
	\end{cases} 
	\end{eqnarray*}
We distinguish the following cases.
\begin{itemize}
    \item Case $k > 5t/2$ \\
One gets 
\begin{eqnarray*}  
    \displaystyle\sum_{P\in Sing(\mathcal{A})}I_{P,max}(\mathcal{A}) &\leq& (q^{2k-t}-q^{4t})q^{t/2}+q^{4t}\cdot\frac{(q^{t/2}+1)^{2}}{4}+(q^{\gcd(k,t)}+1)\frac{q^{4t}}{4} \\
&\leq&q^{2k-{t/2}}-q^{9t/2}+\frac{q^{5t}}{4}+\frac{q^{9t/2}}{2}+\frac{q^{4t}}{4}+\frac{q^{5t}}{4}+\frac{q^{4t}}{4}\\&=&q^{2k-{t/2}}-\frac{q^{9t/2}}{2}+\frac{q^{5t}}{2}+\frac{q^{4t}}{2}\\
    &\leq& q^{2k-{t/2}}+\frac{q^{5t}}{2}\leq
    \frac{3}{2}q^{\max\{2k-{t/2},5t\}} \leq \frac{2}{9}q^{2k} < \frac{2}{9}(q^k+q^{2t}-(q^2+q+1))^2,
\end{eqnarray*}
for $t\geq 4$ and $q \geq 3,$ or $t=2$ and $q\geq 7.$ 
\item Case $k = 5t/2$ \\
We have 
 \begin{eqnarray*}
 \displaystyle\sum_{P\in Sing(\mathcal{A})}I_{P,max}(\mathcal{A}) &\leq& (q^{4t}-q^{3t})q^{t/2}+q^{3t}\cdot\frac{(q^{t/2}+1)^{2}}{4}+(q^{t}+1)\frac{q^{4t}}{4} \\
    &\leq& \frac{q^{5t}}{4}+q^{9t/2}-\frac{q^{7t/2}}{2}+\frac{q^{4t}}{2}+\frac{q^{3t}}{4} \leq \frac{q^{5t}}{4}+q^{9t/2} + \frac{q^{4t}}{2} \\ &\leq&\frac{3}{4}q^{5t}+\frac{q^{4t}}{2} \leq \frac{2}{9}q^{6t} 
    < \frac{2}{9}(q^{3t}+q^{2t}-(q^2+q+1))^2,
 \end{eqnarray*}
 for $q^t\geq 4,$ so for all $q,t.$ \\
 \item Case $2t < k < 5t/2$ \\
 We have
 $$(q^{2k-t}-q^{3t})q^{t/2}+q^{3t}\cdot\frac{(q^{t/2}+1)^{2}}{4}+(q^{s}+1)\frac{q^{4t}}{4} < \frac{2}{9}(q^{k}+q^{2t}-(q^2+q+1))^2$$ where $s:=\gcd(k,t),$ $s \leq t/2,$ $k=\Tilde{k}s,$ $t=\Tilde{t}s,$ $\gcd(\Tilde{k},\Tilde{t})=1.$ \\
 Of course $k > 2t $ implies $\Tilde{k} \geq 2 \Tilde{t}+1,$ so $s+4t=(1+4\Tilde{t})s < (1+2\Tilde{t})2s \leq 2\Tilde{k}s=2k.$ \\
It follows that 
 \begin{eqnarray*}
 \displaystyle\sum_{P\in Sing(\mathcal{A})}I_{P,max}(\mathcal{A}) &\leq&
     (q^{2k-t}-q^{3t})q^{t/2}+q^{3t}\cdot\frac{(q^{t/2}+1)^{2}}{4}+(q^{s}+1)\frac{q^{4t}}{4}\\
     &=& q^{2k-{t/2}}+\frac{q^{4t}}{2}-\frac{q^{7t/2}}{2}+\frac{q^{3t}}{4}+\frac{q^{4t+s}}{4}\\
     &\leq& q^{2k-{t/2}}+\frac{q^{4t+s}}{2} \leq \left\{
	\begin{array}{l}
	\frac{3}{2}q^{2k-1} \ \ \  \ \ \ \ \ \ \ \ \ \textnormal{for} \ t\geq 2\\ 
		(1+q/2)q^{2k-2} \ \ \ \textnormal{for} \ t\geq 4  \\
  (1+q^2/2)q^{2k-2} \ \ \textnormal{for} \ t\geq 6
	\end{array}
	\right. \\
 &\leq&  \frac{2}{9}q^{2k} < \frac{2}{9}(q^{k}+q^{2t}-(q^2+q+1))^2,
 \end{eqnarray*}
 for $t\geq 6$ and $q\geq 3,$ $t\geq 4$ and $q\geq 4$  or $t=2$ and $q\geq 7.$ 
 
\end{itemize}
\end{proof}
\subsection{Main theorem}
Now we are in position to prove our main result.
\begin{thm} \label{thm finale}
    Let $\mathcal{C}=\langle x^{q^t}, x+\delta x^{q^{2t}}, G(x) \rangle \subseteq \mathcal{L}_{n,q}$ be an exceptional $3$-dimensional $\mathbb{F}_{q^n}$-linear MRD code, where $t$ is the minimum integer such that $x^{q^t}\in \mathcal{C}.$  If $(t,q)\notin \{(1,3);(1,4);(1,5);(2,3);(2,4);(2,5),(4,3)\},$  then $\deg_q(G(x))<2t.$
\end{thm}
\begin{proof}
By Theorem \ref{thm collegamento MRD moore set} and Proposition \ref{prop gradi minimi in progressione} it follows that $\min\deg_q(G(X))=2t$ or $t/2.$ \\
Suppose by way of contradiction that $\deg_q(G(X))>2t;$ then $G(X)$ contains at least two terms (see Equations \ref{eq g caso 2t}, \ref{eq G caso t mezzi}). Let us consider $\mathcal{W},\lambda, \mathcal{A}$ as in Equations \ref{equation W}, \ref{lambda scelto opportunamente}, \ref{eq curva A}. Theorems \ref{th main caso 2t} and \ref{thm A componente caso t mezzi} ensure the existence of an absolutely irreducible component defined over $\mathbb{F}_{q^n}$ of $\mathcal{A}$ not contained in the curve defined by $F_{(x,x^q,x^{q^2})}(X,Y,\lambda)=0.$ By Lemma \ref{lem specializzo z in lambda},  $\mathcal{W}$ has an absolutely irreducible $\mathbb{F}_{q^n}$-rational component not contained in $\mathcal{V}.$ Therefore Theorem \ref{thm cafure matera} guarantees the existence of $\mathbb{F}_{q^n}$-rational points in $\mathcal{W}\setminus \mathcal{V}$ for $n$ large enough, and a contradiction arises from Theorem \ref{thm collegamento MRD moore set} and Proposition \ref{prop collegamento varietà}. 
\end{proof}

\section*{Acknowledgments}
This research was supported by the Italian National Group for Algebraic and Geometric Structures and their Applications (GNSAGA - INdAM).

\end{document}